\documentclass[review]{elsarticle}

\usepackage{longtable}
\usepackage{rotating}

\usepackage{epsfig, amsmath,amssymb,epsf,scalefnt,array,setspace}

\usepackage{times}
\usepackage{epstopdf}

\usepackage{amsfonts}
\usepackage{booktabs}
\usepackage{multirow}
\usepackage{multicol}
\usepackage{algorithm}
\usepackage{algorithmic}
\usepackage{stfloats}
\usepackage{bm}
\usepackage{graphicx}
\usepackage{color}
\usepackage{enumerate}

\usepackage[final]{pdfpages}
\usepackage{fancyhdr}

\usepackage{lineno,hyperref}

\modulolinenumbers[5]

\journal{Journal of \LaTeX\ Templates}

\def\cA{{\mathcal{A}}}

 \def\cN{{\mathcal{N}}}

\def\ba{{\mathbf{a}}}    
    
   \def\bn{{\mathbf{n}}} 
    
\def\bu{{\mathbf{u}}} \def\bv{{\mathbf{v}}}   \def\by{{\mathbf{y}}}

\def\bA{{\mathbf{A}}} \def\bB{{\mathbf{B}}} \def\bC{{\mathbf{C}}} \def\bD{{\mathbf{D}}} 
\def\bF{{\mathbf{F}}} \def\bG{{\mathbf{G}}}  \def\bI{{\mathbf{I}}} 
 \def\bL{{\mathbf{L}}}   
\def\bP{{\mathbf{P}}} \def\bQ{{\mathbf{Q}}} \def\bR{{\mathbf{R}}} \def\bS{{\mathbf{S}}} 
\def\bU{{\mathbf{U}}} \def\bV{{\mathbf{V}}} \def\bW{{\mathbf{W}}} \def\bX{{\mathbf{X}}} \def\bY{{\mathbf{Y}}}
\def\bZ{{\mathbf{Z}}}

\def\argmin{\mathop{\mathrm{argmin}}}

\def\tr{\mathop{\mathrm{tr}}}

\def\rank{\mathop{\mathrm{rank}}}
\def\span{\mathop{\mathrm{span}}}
\def\vec{\mathop{\mathrm{vec}}}

\def\diag{\mathop{\mathrm{diag}}}

     \def\d4{\!\!\!\!}

\def\bSig{\mathbf{\Sigma}}   \def\bLam{\mathbf{\Lambda}}
    \def\lam{\lambda}

  \def\R{{\mathbb{R}}}    \def\E{{\mathbb{E}}}

  \def\la{\left|}     \def\ra{\right|}    \def\lA{\left\|}     \def\rA{\right\|}

  \def\sig{\sigma}

  \def\-{\! - \!}  \def\+{\! + \!}  \def\={\! = \!}  \def\>{\! > \!}

\def \log{\mathrm{log}}

\newtheorem{lemma}{Lemma}

\newtheorem{remark}{Remark}

\newcommand{\bef}{\begin{figure}}
\newcommand{\eef}{\end{figure}}
\newcommand{\beq}{\begin{eqnarray}}
\newcommand{\eeq}{\end{eqnarray}}

\newenvironment{proof}[1][Proof]{\begin{trivlist}
\item[\hskip \labelsep {\bfseries #1}]}{\end{trivlist}}

\renewcommand{\qed}{\nobreak \ifvmode \relax \else
\ifdim\lastskip<1.5em \hskip-\lastskip \hskip1.5em plus0em
minus0.5em \fi \nobreak \vrule height0.5em width0.5em
depth0.25em\fi}

\newcommand{\redd}[1]{\textcolor{black}{#1}}
\begin{document}

\begin{frontmatter}

\title{Leveraging Subspace Information for Low-Rank Matrix Reconstruction}


\author[HK]{Wei Zhang \corref{mycorrespondingauthor}}
\cortext[mycorrespondingauthor]{Corresponding author}
\ead{wzhang237-c@my.cityu.edu.hk}
\author[KU]{Taejoon Kim}
\ead{taejoonkim@ku.edu}
\author[KU]{Guojun Xiong}
\ead{gjxiong@ku.edu}
\author[HK,SC]{Shu-Hung Leung}
\ead{eeeugshl@cityu.edu.hk}


\address[HK]{Department of Electronic Engineering, City University of Hong Kong, Kowloon, Hong Kong, China}
\address[KU]{Department of Electrical Engineering and Computer Science, The University of Kansas, Lawrence, KS 66045, USA}
\address[SC]{State Key Laboratory of Terahertz and Millimeter Waves, City University of Hong Kong, Hong Kong, China}

\begin{abstract}
The problem of low-rank matrix reconstruction arises in various applications in communications and signal processing. The state of the art research largely focuses on the recovery techniques that utilize affine maps satisfying the restricted isometry property (RIP). However, the affine map design and reconstruction under a priori information, i.e., column or row subspace information, has not been thoroughly investigated. To this end, we present designs of affine maps and reconstruction algorithms that fully exploit the low-rank matrix subspace information. Compared to the randomly generated affine map, the proposed affine map design permits an enhanced reconstruction. In addition, we derive an optimal representation of low-rank matrices, which is exploited to optimize the rank and subspace of the estimate by adapting them to the noise level in order to achieve the minimum mean square error (MSE). Moreover, in the case when the subspace information is not a priori available, we propose a two-step algorithm, where, in the first step, it estimates the column subspace of a low-rank matrix, and in the second step, it exploits the estimated information to complete the reconstruction. The simulation results show that the proposed algorithm achieves robust performance with much lower complexity than existing reconstruction algorithms.
\end{abstract}


\end{frontmatter}


\section{Introduction}
The size of matrices that appear in modern science and engineering is steadily getting large, which has required heavy computing power to timely process them.
Fortunately, in many cases, data matrices have the low-rank property.
Since the low-rank matrices have much lower degrees of freedom than the full-rank matrices, it is possible to recover them by taking limited observations  \cite{SpareRecovery} \cite{liu2009interior}.
The limited observation refers to the situation when the number of observations involved is much smaller than the dimension of the matrix.

The low-rank matrix reconstruction is to recover a large-dimensional low-rank matrix $\bL\in\R^{M\times N}$ from an affine observation $\cA(\bL)$, where $\cA(\cdot): \R^{M\times N} \mapsto \R^{p}$  \cite{recht}, where $p$ is the number of observations.
Early work in low-rank matrix reconstruction focuses on non-adaptive approaches \cite{recht,chen2015exact, cai2015rop,AlternatingMin,candes2011tight,database,ZhangSD}, where the affine map $\cA(\cdot)$ is fixed during the reconstruction. 
To provide a reliability guarantee, the affine map $\cA(\cdot)$ is designed to satisfy the restricted isometry property (RIP) \cite{recht,chen2015exact, cai2015rop,AlternatingMin,candes2011tight}.
Popularly studied non-adaptive methods are nuclear norm minimization (NNM) \cite{recht,chen2015exact,cai2015rop}, matrix factorization (MF) \cite{AlternatingMin}, and iterative thresholding \cite{SVP, IHT_matrix}.
With convex relaxation, NNM will achieve a reliable reconstruction when the RIP holds \cite{recht,chen2015exact,cai2015rop}. However, the main issue is that NNM does not scale with the problem dimension.
Alternatively, other non-convex techniques \cite{AlternatingMin, SVP, IHT_matrix} were investigated to address the complexity issue introduced by NNM. In particular, the MF in \cite{AlternatingMin} is a block-coordinate descent method, where the low-rank matrix is expressed as the product of  two low-dimensional matrices.
In \cite{SVP, IHT_matrix}, the low-rank solution is guaranteed by taking singular value projection of each iteration result.
Though these non-convex techniques have low complexity, they are still limited by an increased number of observations and many iterations to  converge.

Recent work has studied adaptive low-rank matrix reconstruction techniques \cite{subspaceAware,HighDimension} that significantly reduce the computational complexity as well as enhance the reconstruction accuracy.
Unlike the non-adaptive approaches \cite{recht,chen2015exact, cai2015rop,AlternatingMin,candes2011tight,database,ZhangSD} where they exploit the entire observations $\cA(\bL)$ for reconstruction,  the methods proposed in \cite{subspaceAware,HighDimension} divide the reconstruction task into two steps, where each step only needs to manipulate partial observations of $\cA(\bL)$ to obtain partial information of $\bL$.
Considering that the dimension of partial observations is much smaller than the entire observations, the computational complexity can be reduced substantially.
An interesting observation is that any partial information of matrix subspace provides a significant clue when adapting the affine map to enhance the reconstruction accuracy.
However, how to adapt the affine map to the available subspace information has not been thoroughly studied in the existing literature (e.g., \cite{subspaceAware,HighDimension}).

In this paper, we analyze and design low-rank matrix reconstruction algorithms that best exploit available subspace information of a low-rank matrix to enhance the reconstruction accuracy in two aspects:
(1) we find the optimal representation of the low-rank matrix estimate and (2) we provide the optimal design of the affine map $\cA(\cdot)$.
By the optimal representation, we mean that the degrees of freedom of the subspace estimate, i.e., the rank of the estimate, is adapted to obtain the minimum mean squared error (MSE).
By the optimal design of the affine map $\cA(\cdot)$, we mean that we provide an optimality condition of $\cA(\cdot)$ so the estimate achieves the minimum MSE.
Compared with randomly generated, conventional affine map, the optimized $\cA(\cdot)$ is shown to capture specific information of low-rank matrices.

In practice, the subspace information of a low-rank matrix must be estimated.
Therefore, we propose a two-step method, where the affine map $\cA(\cdot)$ is decomposed into two parts, $\cA_1(\cdot)$ and $\cA_2(\cdot)$.
In the first step, the affine map $\cA_1(\cdot)$ is designed to capture the sparse column subspace information of  $\bL$, and the subspace estimation accuracy is analyzed.
For the second step, different from the existing works in \cite{subspaceAware,HighDimension} where the affine map $\cA_2(\cdot)$  is fixed, the affine map $\cA_2(\cdot)$ in the proposed  algorithm is constructed based on the estimated column subspace information from the first step.
By doing so, the observations $\cA_2(\bL)$ collect the coefficient matrix that represents the optimal combining weights of the estimated column subspace.
An important implication is that the proposed two-step method will not waste the number of observations by not sampling the orthogonal complement of the estimated column subspace.
As demonstrated in our simulation studies, this has its own merit when estimating a large-scale, but substantially low-rank, matrix subject to a limited number of observations.

The paper is organized as follows. Section \ref{background} presents the traditional techniques for low-rank matrix reconstruction and the motivation of our work. In Section \ref{algorithm}, the low-rank matrix reconstruction under the column subspace information is investigated, where the optimization for the affine map and the minimal matrix representation is discussed.
In Section \ref{two step}, we introduce the proposed two-step low-rank matrix reconstruction technique, and analyze its error bound.
The simulation results are illustrated and analyzed in Section \ref{simulation}. Finally, conclusions are drawn in Section \ref{conclusion}.
\subsubsection*{Notations}
A bold lower case letter $\ba$ is a vector and a bold capital letter $\bA$ is a matrix.
$\bA^T,\bA^{\!-1}$, tr($\bA$), $| \bA  |$, $\| \bA \|_F$, $\| \bA  \|_*$, and $\|\ba\|_2$ are, respectively,    the transpose, inverse, trace, determinant, Frobenius norm, nuclear norm (i.e., the sum of the singular values of $\bA$) of $\bA$, and $l_2$-norm of $\ba$.
$[\bA]_{:,i}$, $[\bA]_{i,:}$, $[\bA]_{i,j}$, $[\ba]_i$ are, respectively, the $i$th column, $i$th row, $i$th row and $j$th column entry of $\bA$, and $i$th entry of vector $\ba$.
$\vec(\bA)$ stacks the columns of $\bA$ and form a long column vector.
$\diag(\bA)$ extracts the diagonal entries of $\bA$ to form a column vector.
$\bI_M \! \in \! \R^{M\times M}$ is the identity matrix.
$\text{col}(\bA)$ denotes the subspace spanned by the columns of matrix $\bA$.

\section{Low-rank matrix reconstruction} \label{background}
In this section, we define and review the low-rank matrix reconstruction problem, and provide a general motivation of our work.
\subsection{Classical techniques: non-adaptive}
Suppose a matrix observation model where the observations of $\bL\in\R^{M\times N}$ are made through the affine map $\cA(\cdot):\R^{M\times N}  \rightarrow \R^{p}$, where $p$ is the number of observations and $\rank(\bL)=r \ll \min(M,N)$.
Here, the $i$th element of the affine map $\cA(\bL)$ is given by
\beq
[\cA(\bL)]_i = \tr(\bX^T_i \bL), ~i=1,2\ldots p, \label{affine map form}
\eeq
where $\bX_i \in \R^{M \times N }$ is a measurement matrix associated with the $i$th entry of $\cA(\bL)$. In practice,  the observations are corrupted by noise
\beq
\by = \cA(\bL) + \bn \label{observations},
\eeq
where $\bn \in \R^{p }$ is the additive white Gaussian noise vector with zero mean and $\bC \in \R^{p \times p}$ covariance, i.e., $\bn \sim \cN(\mathbf{0}, \bC)$.

When the number of observations $p$ is not sufficiently large, recovering the matrix $\bL$ from $\by$ is not always possible.
Nonetheless, when the $\bL$ is low-rank, it is possible to reconstruct $\bL$ within a certain accuracy  \cite{recht}.

The NNM  formulation of the low-rank matrix reconstruction problem \cite{recht,overview,chen2015exact,cai2015rop} can be given by,
\beq
\min \limits _{\bL}\lA \bL \rA_* +\tau \lA \by - \cA(\bL)\rA_F^2 , \label{proximal gradient}
\eeq
where $\tau$ is a parameter which balances the values of two parts in the objective function.
This problem is an unconstrained convex optimization program, and can be solved  by using a proximal gradient method \cite{proximal,combettes2011proximal}.
The main drawback is that it does not scale to large-dimensional matrices due to the high complexity of singular value decomposition (SVD) in each iteration  \cite{overview, cai2010}.
\redd{Though the computational complexity can be resolved by using accelerated proximal gradient method and in-exact SVD \cite{accYao,YaoNonCVX}, the number of required observations for NNM to recover a rank-$r$ matrix $\bL \in \R^{M \times N} (M \le N)$ is still high, i.e., $\mathcal{O}(rN \log(N))$ \cite{cai2010, accYao}. When the number of available observations is less than this requirement, the reconstruction accuracy will deteriorate \cite{ZhangSD}.}

The development of MF \cite{koren2009matrix,Haldar_factor,AlternatingMin} has provided an efficient technique to handle large-scale low-rank matrices.
The MF problem for low-rank matrix reconstruction can be  formulated as
\beq
\begin{aligned}
\big( \bB^\star, \bR^\star \big) = \argmin \limits_ {\bB\in \R^{M \times r},\bR\in \R^{N \times r}} \lA \by - \cA( \bB \bR^T) \rA_2^2.  \label{matrix factor}
\end{aligned}
\eeq
Solving \eqref{matrix factor} results in $\widehat{\bL} = \bB^\star(\bR^\star)^T$.
The problem in \eqref{matrix factor} is non-convex, however, the alternating minimization with power factorization \cite{Haldar_factor} can be used to obtain an effective sub-optimal solution.
The complexity of the MF is much lower than that of NNM, and MF outperforms NNM when the matrix is substantially low-rank \cite{Haldar_factor, AlternatingMin,ZhangSparse}.
A critical drawback of MF is that it can not guarantee the convergence to the true low-rank matrix, which is largely dependent on the initializations of $\bB$ and $\bR$ \cite{AlternatingMin}. Moreover, since each iteration of solving MF is a least square, it still requires much computational time when the dimension of $\bL$ is very large.

\subsection{Previous work: adaptive techniques}
In order to reduce the computational complexity of low-rank matrix reconstruction, the works in  \cite{subspaceAware,HighDimension} have proposed adaptive matrix sensing techniques.
Specifically, in stead of manipulating the whole observations $\by$ in \eqref{observations}, they process the partial observations in two stages,
\beq
\begin{bmatrix}
  \by_1 \\
 \by_2
\end{bmatrix}
 =
\begin{bmatrix}
  \cA_1(\bL) \\
 \cA_2(\bL)
\end{bmatrix}
 + \begin{bmatrix}
  \bn_1 \\
 \bn_2
\end{bmatrix}, \label{part observation}
\eeq
where $\by_1 \in \R^{p_1 }$, $\by_2 \in \R^{p_2 }$, and $p = p_1 + p_2$.
In the first stage, $\by_1 = \cA_1(\bL) + \bn_1$ is taken to estimate the column subspace of the $\bL$, which we denote $\widehat{\bF} \in \R^{M \times r}$.
The second part $\by_2 = \cA_2(\bL) + \bn_2$ is to estimate the coefficient matrix with respect to the estimated $\widehat{\bF}$ in order to reconstruct $\bL$.
Each technique in \cite{subspaceAware, HighDimension} proposes a method of generating $\cA_2(\cdot)$ in \eqref{part observation}.
In particular, in \cite{subspaceAware}, the elements in $\cA_2(\cdot)$ are drawn from i.i.d.  Gaussian distributions.
While in \cite{HighDimension}, $\cA_2(\cdot)$ consists of elements vectors to select $m_2~(>r)$ rows of $\bL$.

Since $\by_1$ and $\by_2$ are of much smaller number than $p$,  \cite{subspaceAware, HighDimension} provide the ways to utilize the subspace information to assist the reconstruction of $\bL$, and reduce the computational complexity as well.
However, these techniques did not fully consider the design of affine maps $(\cA_1(\cdot), \cA_2(\cdot)) $ and the optimal representation of the estimated low-rank matrix, which can potentially improve the reconstruction accuracy.
Overall, how to adapt the affine maps to the available subspace information has not been thoroughly studied in the existing literature.

\subsection{Motivation}

The RIP is an important characterization that provides the performance guarantees for reliable low-rank matrix reconstruction. Because of this, it has been widely employed \cite{recht, AlternatingMin,overview,chen2015exact,cai2015rop}.
Under the RIP, many of the state of the art reconstruction methods
\cite{recht,chen2015exact,cai2015rop,AlternatingMin, SVP, IHT_matrix,subspaceAware,HighDimension} employ randomly generated affine maps and focus on the reconstruction algorithms.
On the other hand, if one can optimize the affine map to lower reconstruction error of a specific algorithm, it is expected that this will reduce the observation overhead.
When the subspace information of a low-rank matrix is partially known, it would be possible to adapt $\cA(\cdot)$ to the known subspace by minimizing MSE.

\section{Affine map design and minimal representation with known column subspace} \label{algorithm}
In this section, we analyze the optimal design of the affine map and provide the optimal representation of the estimate given the column subspace information.

Denote the singular value decomposition (SVD) of $\bL$ as
\beq
\bL = \bU \bLam \bV^T=\sum_{k=1}^r \lambda_k \bu_k \bv_k^T \label{svd c}
\eeq
where $\bU =[\bu_1,\bu_2,\cdots,\bu_r]\in \R^{M\times r}$ and $\bV =[\bv_1,\bv_2,\cdots,\bv_r]\in \R^{N\times r}$ are the left and right singular matrix, respectively. The $\bLam \in \diag([\lam_1,\lam_2,\cdots,\lam_r])$ is a diagonal matrix where the singular values $\lam_k$ are arranged in descending order.

\subsection{Design of affine map }
Denote $\bF \in \R^{M \times r}$ as a semi-unitary matrix, such that $\text{col}(\bF)$ is the column  subspaces of $\bL$.
It is worth noting that $\bF$ is not necessarily same as $\bU$ in \eqref{svd c}.
Given the $\bF$, the low-rank matrix $\bL$ can be expressed as
\beq
\bL = \bF \bQ, \label{form1 c}
\eeq
where $\bQ \in \R^{r \times N}$ is a combining matrix.
Given the representation in \eqref{form1 c}, we find the estimate of $\bL$ in the form,
$
\widehat{\bL} = \bF \widehat{\bQ} $.
Thus, the problem of estimating $\bL$ is equivalent to reconstruct $\widehat{\bQ}$.

To begin with, we reshape the affine map $\cA$ in \eqref{affine map form} as $\bS \in \R^{p \times MN}$,
\beq
\bS =
\left[
\begin{matrix}
  \vec(\bX_1), \cdots,
 \vec(\bX_p)
\end{matrix}
\right]^T. \label{measurement c}
\eeq
Then the observation $\by$ in (2) can be rewritten as
\beq
\by  &= &\bS \vec(\bL) + \bn \nonumber \\
&= &\bS \vec(\bF \bQ) + \bn \nonumber \\
&= & \bS (\bI_{N} \otimes \bF) \vec(\bQ)+\bn. \label{observation vec c}
\eeq
To formulate a design criterion to optimize the affine map $\bS$, we impose a power constraint $\| \bS \|_F^2 \le P$ to $\bS$.
In \eqref{observation vec c}, we denote
\beq
\bA = \bS ( \bI_N \otimes \bF), \label{reshape1 c}
\eeq
where $\bA \in \R^{p \times Nr}$.
Given the observation model in \eqref{observation vec c}, the efficient estimator\footnote{We call an estimator is efficient, if it is unbiased and attains the Cramer-Rao lower bound \cite{estimationBook}.} of $\bQ$ yields \cite{estimationBook}
\beq
\vec(\widehat{\bQ}) = (\bA^T \bC^{-1} \bA)^{-1} \bA^T \bC^{-1}\by. \nonumber
\eeq
The MSE of $\widehat{\bL}$ is then written as
\beq
\E\left(\lA  \widehat{\bL} - \bL\rA_F^2\right) \!\!\!\!&=&\!\!\!\!\E\left( \lA  \widehat{\bQ} - \bQ\rA _F^2\right) \nonumber\\
&=& \!\!\!\!\E\left( \lA(\bA^T \bC^{-1} \bA)^{-1} \bA^T  \bC^{-1} \bn \rA _F^2\right) \nonumber\\
& = & \!\!\!\! \tr(\bA^T\bC^{-1}\bA)^{-1}. \label{MSE form c}
\eeq

Now the problem is optimizing the affine map $\bS$ by minimizing the MSE in \eqref{MSE form c} subject to \eqref{reshape1 c} and the power constraint and it can formally be written as
\beq
&\min\limits_\bS \tr(\bA^T\bC^{-1}\bA)^{-1} \nonumber \\
&\text{subject to~~}  \bA = \bS(\bI_N \otimes \bF),~\lA \bS \rA_F^2 \le P.
 \label{optimal training1 c}
\eeq
The following lemma shows that the problem in \eqref{optimal training1 c} is equivalent to a more simple problem.
\begin{lemma} \label{lemma1}
Denote the optimal solution to \eqref{optimal training1 c} as $\widehat{\bS}$. Suppose the following problem
\beq
\min\limits_\bA \tr(\bA^T\bC^{-1}\bA)^{-1} ~ \text{subject to~} {\tr(\bA^T \bA)} \le P \label{general ray1t c}
\eeq
has the minimal solution as $\widehat{\bA}$. Then, the equality holds,
\beq
\widehat{\bS} = \widehat{\bA}(\bI_N \otimes \bF)^T. \label{eq_con11 c}
\eeq
\end{lemma}
\begin{proof}
This is a direct consequence of the equality $(\bI_N \otimes \bF)^T (\bI_N \otimes \bF)=\bI_{Nr}$. \hfill \qed
\end{proof}

Lemma \ref{lemma1} specifies that in order to obtain the optimal $\bS$, we can turn to solving the problem \eqref{general ray1t c}, and the optimal $\widehat{\bS}$ can be calculated according to \eqref{eq_con11 c}.
The solution of \eqref{general ray1t c} is summarized in the following lemma.

\begin{lemma} \label{augL}
Suppose the power-constrained MSE minimization problem in \eqref{general ray1t c}, i.e.,
\beq
\min\limits_\bA \tr(\bA^T\bC^{-1}\bA)^{-1} ~ \text{subject to~} {\tr(\bA^T \bA)}\le P. \label{general ray1 c}
\eeq
Denote the SVD of $\bA$ and $\bC$ as $\bA = \bU_A \bSig_A \bV_A^T$ with $\bU_A \in \R^{p \times Nr }$,
$\bSig_A \in \R^{Nr \times Nr }$, $\bV_A \in \R^{Nr \times Nr }$
and $\bC = \bU_C \bSig_C \bU_C^T$ with $\bU_C \in \R^{p \times p }$,
$\bSig_C \in \R^{p \times p }$, respectively.

Then, the optimal solution to \eqref{general ray1 c} has the following form,
\beq
\begin{split}
&\bU_A = [\bU_C]_{p-Nr+1:p},\\
&\bD  =[\bSig_C]_{p-Nr+1:p,p-Nr+1:p},\\
&\mu = \tr(\bD^{    \frac{1}{2}})^2/P^2,\\
&\bSig_A^2 = \mu^{-\frac{1}{2}} \bD^{\frac{1}{2}},~
\bV_A^T \bV_A = \bI_{Nr}. \label{general optimal condition1 c}
\end{split}
\eeq
In addition, the minimal MSE is given by
\beq
\frac{1}{P}\left(\sum_{k=p-Nr+1}^{p} \lambda_{c,k}^{\frac{1}{2}}\right)^2,
\eeq
where $\lambda_{c,k}$ is the $k$th largest eigenvalue of $\bC$.
After $\widehat{\bA}$ is designed based on \eqref{general optimal condition1 c}, the $\widehat{\bS}$ is obtained by \eqref{eq_con11 c}.
\end{lemma}

\begin{proof}
The problem \eqref{general ray1 c} is a standard differentiable convex optimization problem. This can be optimally solved by applying Karush-Kuhn-Tucker (KKT) conditions \cite{boyd}. We omit the details because it is a standard procedure. \hfill \qed
\end{proof}

\subsection{Optimal representation of an estimator}
In the previous section, we assume that the estimate of $\bL$ has the form $\widehat{\bL} = \bF \widehat{\bQ}$.
In this subsection, we investigate the optimal representation of the low-rank matrix estimate given the noisy affine map in (2), which can achieve the minimum MSE.

The following lemma shows that the estimation error will not increase if we project any estimation result onto the column subspace of the low-rank matrix.

\begin{lemma} \label{lemma2}
Suppose that $\bL \in \R^{M \times N}$ is a low-rank matrix and has a decomposition $\bL  = \bF \bQ$, where $\bF \in \R^{M \times r}$ with $\bF^T \bF = \bI_r$.
Given the observation $\by = \cA(\bL) + \bn$ with $\E(\bn \bn^T) = \bC$,
the following inequality holds
\beq
\lA \tilde{\bL}  - \bL \rA_F^2 \ge   \lA \bF \bF^T \tilde{\bL} - \bL \rA_F^2,
\eeq
where $\tilde{\bL}$ is the arbitrary estimate of $\bL$.
\end{lemma}
\begin{proof}
It is worth noting that $\bF \bF^T$ denotes the projection matrix for the column subspace of $\bL$.
Let us first define $\bP_{\bF} = \bF\bF^T$, $\bP_{\bF_\perp} =\bI_M - \bP_{\bF}$,
then the estimation error of $\tilde{\bL}$ is given by
\beq
\begin{split}
\lA \tilde{\bL}  - \bL \rA_F^2 &= \lA \bP_{\bF} \tilde{\bL}   +   \bP_{\bF_\perp} \tilde{\bL} - \!\! \bL \rA_F^2\\
& =  \lA \bP_{\bF} \tilde{\bL}  - \bL \rA_F^2 \!\!+\!\!\lA \bP_{\bF_\perp} \tilde{\bL} \rA _F^2\\
& \ge \lA \bP_{\bF} \tilde{\bL}  - \bL \rA_F^2. \nonumber
\end{split}
\eeq
This concludes the proof. \hfill \qed
\end{proof}

The lemma concludes that the optimal estimation should have the form $\widehat{\bL} = \bP_{\bF} \tilde{\bL} $ in order to further minimize the estimation error. In other words, the column subspace of $\widehat{\bL}$ should lie on the column subspace of $\bL$,
$\text{col} (\bF_1) \subseteq \text{col} (\bF)$
where we denote the column subspace of estimate $\widehat{\bL}$ as $\text{col}(\bF_1)$ and $\bF_1 \in \R^{M \times d}$ with $\bF_1^T \bF_1 =\bI_d$, $d\le r$.

Given $\bF_1$, we rewrite the true low-rank matrix in \eqref{form1 c} as
\beq
\bL = \bF \bQ
=  \begin{bmatrix}
        \bF_1 & \bF_2
      \end{bmatrix}
  \begin{bmatrix}
    \bQ_1  \\
   \bQ_2
  \end{bmatrix}
   \label{matrix sep c}
\eeq
where $\bF = [\bF_1, \bF_2]$ is semi unitary $\bF^T \bF = \bI_r$ and $\bQ_1 \in \R^{d \times N}$, and $\bQ_2 \in \R^{(r-d) \times N}$.
With the column subspace of $\widehat{\bL}$ being $\text{col}(\bF_1)$,  we obtains the estimate as
\beq
\widehat{\bL} = \bF_1 \widehat{\bQ}_1,
\eeq
where $\widehat{\bQ}_1 \in \R^{d \times N}$. Then, the observation $\by$ in \eqref{observations} satisfies the following
\beq
\by &=& \cA(\bF_1 {\bQ}_{1} ) +  \cA(\bL - \bF_1 {\bQ}_{1} ) + \bn \nonumber \\
&=&\cA(\bF_1 {\bQ}_{1} ) +   \tilde{\bn}\nonumber \\
&=& \bS(\bI \otimes \bF_1)\vec(\bQ_{1}) + \tilde{\bn}\nonumber \\
&=& \bA_1 \vec(\bQ_1) + \tilde{\bn}, \label{observation2 c}
\eeq
where $\tilde{\bn} =\cA(\bL - \bF_1 {\bQ}_{1} ) + \bn $, and $\bA_1 = \bS(\bI \otimes \bF_1) \in \R^{p \times dN}$.
The effective noise covariance can be calculated as
\beq
\bC_1 = \E(\tilde{\bn}\tilde{\bn}^T) = \bA_1 \vec(\bQ_1) \vec(\bQ_1)^T  \bA_1^T  + \bC. \label{new cov}
\eeq
Given the observation \eqref{observation2 c}, the efficient estimator of ${\bQ}_1$ yields
\beq
\vec(\widehat{\bQ}_{1})=(\bA_1^T \bC_1^{-1}\bA_1)^{-1} \bA_1^T  \bC_1^{-1}\by.
\eeq
 The MSE of $\widehat{\bL}$ is given by
\beq
\!\!\!\!\!\!\!\!\!\!\!\!\!\!\!\!\E\left(\lA  \widehat{\bL} - \bL\rA_F^2\right) \!\!\!\!\!\!\!&=&\!\!\!\!\E\left( \lA \bF_1 \widehat{\bQ}_{1}  - \bL\rA_F^2 \right) \nonumber \\
\!\!\!\!\!\!\!\!&=&\!\!\!\!\E\left( \lA  \widehat{\bQ}_1  \!\!- \!\! \bQ_1\rA _F^2 \!\!\right) + \! \lA \bL \rA_F^2   \!- \! \lA {\bQ}_{1}  \rA_F^2\nonumber\\
\!\!\!\!\!\!\!\!&=& \!\!\!\!\tr(\bA_1^T \bC_1^{-1} \bA_1)^{-1}   \!+ \! \lA \bL \rA_F^2   \!\!- \!\! \lA {\bQ}_{1}  \rA_F^2.
\eeq
In order to minimize the MSE above, we solve the following optimization problem in terms of the affine map,
\beq
\!\!\!\! \min \limits_{\bS,\bF_1}  \tr(\bA_1^T \bC_1^{-1} \bA_1)^{-1}  \!\!+ \!\! \lA \bL \rA_F^2   \!- \! \lA {\bQ}_1  \rA_F^2 \nonumber \\
\text{subject to } \bA_1 = \bS (\bI \otimes \bF_1),~ \lA \bS \rA_F^2 \le P.
\label{MSE expression A c}
\eeq

The following Lemma analyzes optimality conditions for  $\bF_1$ and $\bS$ to achieves the minimal MSE.

\begin{lemma} \label{lemma4}
The optimal $\bF_1 \in \R^{M \times d}$  and  $\bS$ that achieve the minimal MSE in \eqref{MSE expression A c} satisfy the following conditions:

(1) The optimal rank of the estimate satisfies
\beq
d_{opt}=\argmin \limits _d \left( \frac{1}{P}\left(\sum_{k=p-Nd+1}^{p} \lambda_{c_1,k}^{\frac{1}{2}}\right)^2 +\sum_{k=d+1}^{r}\lambda_{k}^2 \right), \label{column MSE}
\eeq
where $ \lambda_{c_1,k}$ is the $k$th largest singular value of $\bC_1$, and $ \lambda_{k}$ is the $k$th largest singular value of $\bL$.

(2) $\bF_1$ spans the subspace of the dominant $d_{opt}$ left  singular vectors of $\bL$.

(3) The optimal ${\bS}$ is  $\widehat{\bS} = \widehat{\bA}_1 (\bI \otimes \bF_1)^T$, where $\widehat{\bA}_1 \in \R^{p \times dN}$ is the solution of \eqref{general ray1 c} in Lemma \ref{augL} by substituting $\bC$ with  $\bC_1$ in \eqref{new cov}.
\end{lemma}

\begin{proof}
Since there are two optimization variables: $\bF_1 $ and $\bS$, we will solve the whole optimization problem in two steps, where these two steps can be described by using the problem
\beq
\!\!\!\! \min \limits_{\bF_1} \left( \min \limits_\bS\tr(\bA_1^T \bC_1^{-1} \bA_1)^{-1}  \!\!+ \!\! \lA \bL \rA_F^2   \!\!- \!\! \lA {\bQ}_{1}  \rA_F^2\right) \nonumber \\
\text{subject to } \bA_1 = \bS (\bI \otimes \bF_1),~ \lA \bS \rA_F^2 \le P. \label{problem2 c}
\eeq
The first step is to design $\bS$ for any $\bF_1$. The second step is to optimize $\bF_1$ that minimizes the whole objective function.

First of all, for any arbitrary $\bF_1$, the value  $(\lA \bL \rA_F^2  \!\!- \!\! \lA {\bQ}_1 \rA_F^2)$ is constant due to $\bQ_1 = \bF_1^T \bL$, simplifying \eqref{problem2 c} to
\beq
&\min \limits_\bS\tr(\bA_1^T \bC_1^{-1} \bA_1)^{-1}  \nonumber \\
&\text{subject to } \bA_1 = \bS (\bI \otimes \bF_1),~ \lA \bS \rA_F^2 \le P. \label{form2 c}
\eeq
The objective value $\tr(\bA_1^T \bC_1^{-1} \bA_1)^{-1}$ attains the minimum  when $\bS$ is calculated according to Lemma \ref{lemma1} and Lemma \ref{augL}, yielding
\beq
 \frac{1}{P}\left(\sum_{k=p-Nd+1}^{p} \lambda_{c_1,k}^{\frac{1}{2}}\right)^2. \nonumber
\eeq

Now the problem in \eqref{problem2 c} is over the $\bF_1$,
\beq
\min \limits _{\bF_1} \frac{1}{P}\left(\sum_{k=p-Nd+1}^{p} \lambda_{c_1,k}^{\frac{1}{2}}\right)^2+ \lA \bL \rA_F^2  - \lA {\bQ}_1 \rA_F^2. \label{optimize d c}
\eeq
\begin{figure}[t]
\centering
\includegraphics[width=3.5 in]{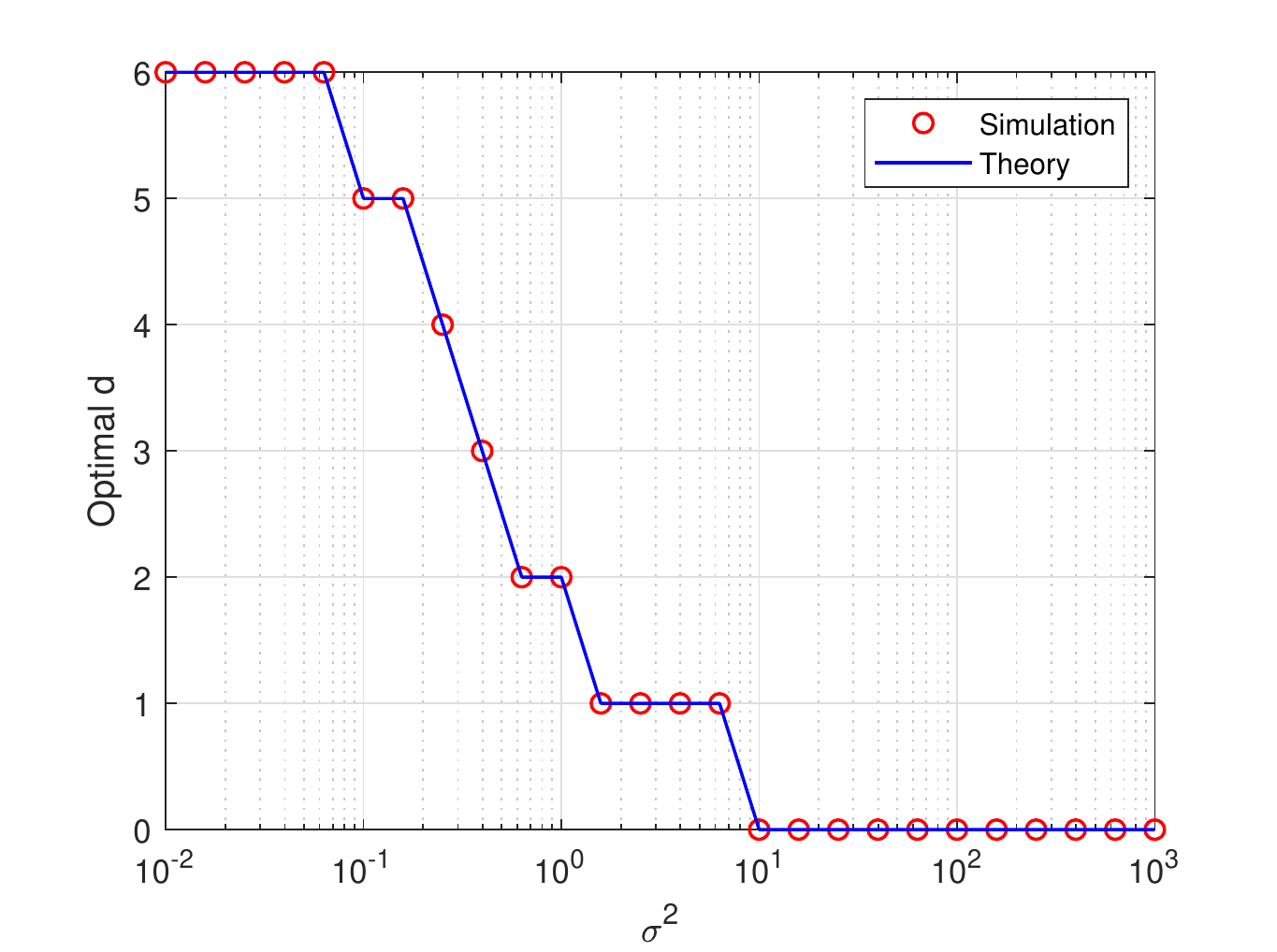}
\caption{ Optimal $d$ vs. Noise levels ($M = 20, N = 50, r=6$).} \label{optimal d c}
\end{figure}

\noindent
It is worth noting that the first part of \eqref{optimize d c} only relies on the dimension of $\bF_1$, i.e., $d$.
In other words, the value of the first part is a constant for any fixed $d$.
However, for any given $d$, the value of the second part (i.e.,$\lA \bL \rA_F^2  - \lA {\bQ}_{1} \rA_F^2$) can be minimized when $\bF_1$ spans the subspace of the dominant $d$ left singular vectors of $\bL$.
Then, we obtain the MSE in terms of $d$ as
\beq
\text{MSE}(d)=\frac{1}{P}\left(\sum_{k=p-Nd+1}^{p} \lambda_{c_1,k}^{\frac{1}{2}}\right)^2 +\sum_{k=d+1}^{r}\lambda_{k}^2.
\label{eq temp3 c}
\eeq
To minimize \eqref{eq temp3 c}, the $d$ should be chosen by
\beq
d_{opt}= \min \limits _d  \text{MSE}(d), \nonumber
\eeq
which concludes the proof. \hfill \qed
\end{proof}

In order to understand the Lemma \ref{lemma4} in detail, we analyze and illustrate the scenario when the noise is i.i.d. Gaussian such as $\E(\bn \bn^T) = \sigma^2\bI_p$. In this case, the MSE over $d$ is expressed as
\beq
\frac{d^2 N^2}{P}\sigma^2 +\sum_{k=d+1}^{r}\lambda_{k}^2,
\label{eq temp2 c}
\eeq
 Because $\lambda_{k}$ is monotonically decreasing with $k$,  the optimal $d$ that minimizes \eqref{eq temp2 c} should satisfy the following
\beq
 d_{opt} &=& \argmin \limits _{d=0,1\cdots,r} \left( \frac{d^2 N^2}{P}\sigma^2 +\sum_{k=d+1}^{r}\lambda_{k}^2\right). \label{con2 c}
\eeq
Observing \eqref{con2 c},  the value of first part is increasing over $d$, while the value of second part is decreasing over $d$.
This means that when the noise level is high, such as a large $\sigma^2$, we will have a small $d_{opt}$ for fixed $P$ and $\{ \lambda_k\}_{k=1}^r$.
When the noise level is low, the value of $d_{opt}$ will approach $r$, which is the rank of matrix $\bL$.

In Fig. \ref{optimal d c}, we plot the optimal values of $d$ which achieve the smallest MSEs for different noise levels $\sigma^2$. The line of theoretical optimal $d$ is calculated according to \eqref{con2 c}. It is clear that the theoretical $d_{opt}$ well matches with the simulation results.
Specifically, when the noise level is low, the $d_{opt}$ is equal to the rank of matrix, i.e., $r$.
When the noise level is high, $d_{opt}$  approaches to zero, which is consistent with our analysis.

\section{Two-step low-rank matrix reconstruction} \label{two step}

As we discussed in the previous section, when the column subspace of low-rank matrix is known in advance, the affine map can be attained as we analyzed in Section \ref{algorithm}.
In practice, however, the column subspace of $\bL$ is often not known as a priori, in which we have to first estimate the column subspace and then using this estimated priori information to finish the estimation of the low-rank matrix.

In this section, we propose a two-step method to reconstruct $\bL$, which consists of column subspace learning and coefficient matrix learning, where we assume the powers for these two steps are $P_1$ and $P_2$ with $P_1+P_2 = P$.

\subsection{Column subspace learning}
For a matrix $\bL \in \R^{M \times N}$ with rank $r$, the column subspace of $\bL$, i.e., $\text{col}(\bL)$, can be expressed by the a semi-unitary matrix $\bU_p \in \R^{M \times r}$ such that $\text{col}(\bL) =\text{col}(\bU_p) $. Here, we denote $\bU_p$ as the column subspace matrix of $\bL$.
It is worth noting that the column subspace matrix of a given $\bL$ is not unique.
In particular, the left singular vectors of $\bL$ is a candidate of column subspace matrix for $\bL$.

Denote $\bL_1 \in \R^{M \times m}$ as the sub-matrix of $\bL$, which is generated by selecting $m$ columns of $\bL$. If there are $r$ columns in $\bL_1$ are linearly independent, we will have $\text{col}(\bL) =\text{col}(\bL_1)$.
In general, however, we can not guarantee that subspace equality between $\bL_1$ and $\bL$.
Fortunately, according to \cite[Lemma $2$]{HighDimension},
if we randomly select $m$ columns of $\bL$ as $\bL_1$, and $m~(\ge r)$ is large enough, the selected columns of matrix $\bL_1$ span the column subspace of $\bL$ with high probability.

In more detail, we let $\bY_1 \in \R^{M \times m}$ be the observations under noise, i.e.,
\beq
\bY_1 =  \bL \bZ_1 + \sqrt{{m M}/{P_1 } } \bW_1  =\bL_1 + \bW_s , \label{expression L1}
\eeq
where $\bZ_1 \in \R^{N \times m}$ is the column sampling matrix, which randomly selects $m$ columns from $\bL$,  and $\bW_1 \in \R^{M \times m}$ is reshaped the noise matrix with $[\bW_1]_{i,j}\sim \cN(0,\sigma^2)$.
Note that the number of observations to obtain $\bY_1$ in the first step is $p_1 = m M$. The remaining question is to obtain the column subspace of $\bL_1$ from the observations $\bY_1$.

Considering that $\bL_1$ is also a low-rank matrix when $m \gg r$, it is intuitive to estimate $\bL_s$ by solving the following nuclear norm and $\ell$-norm minimization problem \cite{HighDimension},
\beq
\min \limits_{\bL_1} \lA \bL_1 \rA_* + \varepsilon\lA \bY_1 - \bL_1 \rA_{\ell} , \label{compact form}
\eeq
where $\ell$ denote different norm operation determined by the type of noise $\bW_s$ in \eqref{compact form}, and  $\varepsilon$ is the trade-off parameter.

However, the column subspace learning approach above requires $m \gg r$, which needs $mM$ observations in total. Alternatively, we put a more straightforward but rather robust method.
Considering the observations $\bY_1 = \bL_1 + \bW_s$, when the elements in $\bW_s$ are sufficiently small, the column subspace of $\bL_1$ can be approximated by the column subspace of $\bY_1$.
Suppose the singular value decomposition of $\bL_1$ is given by,
\beq
\bL_1 = \bU_p \boldsymbol{\Sigma}_p  \bV_p^T, \label{svd of Ls}
\eeq
where the left singular vector matrix $\bU_p \in \mathbb{R}^{M \times r}$, the right singular vector
matrix $\bV_p \in \mathbb{R}^{N \times r}$, and $\boldsymbol{\Sigma}_p$ is a diagonal matrix where the singular values are sorted in a descending order.

Because of the noise matrix $\bW_s$, the matrix $\bY_1$ has full rank with probability one, and the singular value decomposition of $\bY_1$ can be expressed as,
\beq
\bY_1 =\widehat{\bU} \widehat{\boldsymbol{\Sigma}} \widehat{\bV}^T  =
\begin{pmatrix} \widehat{\bU}_p & \widehat{\bU}_o \end{pmatrix}
\begin{pmatrix} \widehat{\boldsymbol{\Sigma}}_p & \bold{0}\\ \bold{0} & \widehat{\boldsymbol{\Sigma}}_o \end{pmatrix}
\begin{pmatrix} \widehat{\bV}_p^T \\ \widehat{\bV}_o^T \end{pmatrix} \label{svd subspace}
\eeq
where $\widehat{\bU}_p \in \mathbb{R}^{M \times r}$,  $\widehat{\bV}_p \in \mathbb{R}^{N \times r}$,
$\widehat{\bU}_o \in \mathbb{R}^{M \times (m-r)}$, $\widehat{\bV}_o \in \mathbb{R}^{N \times (m-r)}$. The diagonal matrices $\widehat{\boldsymbol{\Sigma}}_p$ and $\widehat{\boldsymbol{\Sigma}}_o$ include singular values in a descending order.

When the elements in error matrix $\bW_s$ are small compared with the $\bL_1$, we can treat $\text{col}(\widehat{\bU}_p)$ as the estimation for the column subspace of $\bL$, i.e., $\text{col}({\bU}_p)$ in \eqref{svd of Ls}.
In the following sections, we will discuss the subspace estimation accuracy in more details.

\begin{remark}
When we have no idea about the rank of the matrix $\bL$,
we can estimate it through the observations $\bY_1$ in the first step.
When the matrix $\bL_1$ is zero, the matrix $\bY_1$ only contains noise. Moreover, the singular values of $\bY_1$ are supported on $\left[\sigma \sqrt{M} (1-{\sqrt{m/M}}),\sigma \sqrt{M} (1+{\sqrt{m/M}})\right]$ as $M \rightarrow \infty$ \cite{RandomMatrix}. It inspires us to use the threshold $\tau =\sigma \sqrt{M} (1+{\sqrt{m/M}})$ to determine the rank estimation, namely
\beq
\hat{r} = \max\{i:\hat{\lambda}_i \ge \tau \}, \label{rank estimate}
\eeq
where $\hat{\lambda}_i$ is the $i$th largest singular value of $\bY_1$. Moreover, when all the singular values of $\bY_1$ is smaller that $\tau$, we just let $\hat{r}=0$. In the simulation part, we show that the similar performance can be achieved by using the estimation of rank $\hat{r}$.

\end{remark}

\begin{algorithm} [t]
\caption{Two-step method for reconstruction of low-rank matrices}
\label{alg1}
\begin{algorithmic} [1]
\STATE Input: The data matrix $\bL \in \R^{M \times N}$ and the power $P$.
\STATE Initialization: The column selecting matrix $\bZ_1 \in\R^{ N \times m}$, powers for two steps are $P_1$ and $P_2$ with $P = P_1 +P_2$, and the permutation matrix $\bP \in \R^{N \times N}$.
\STATE Column subspace learning:
\begin{itemize}
  \item Sample the columns of  $\bL$ by $\bY_1 = \bL \bZ_1 + \sqrt{mM/P_1} \bW_1 = \bL_1 + \bW_s$ with noise matrix $\bW_1$.
  \item Calculate the column subspace basis matrix $\widehat{\bU}_p$ through  \eqref{svd subspace} and $\widehat{\bQ}_1$ in \eqref{Q1 estimation}.
\end{itemize}

\STATE Calculate the coefficient matrix:
\begin{itemize}
  \item Let the reshaped affine map of $\cA_2(\cdot)$ be $\bS_2 = \sqrt{\frac{P_2}{r(N-m)} } \bI_{r(N-m)}\left(  \bI \otimes \widehat{{\bU}}_p \right)^T$\!\!\!\!, obtain the observations through $\by_2 = \cA_2 (\bL_2) +\bn_2$.
  \item  Calculate the estimation of coefficient matrix $\widehat{\bQ}_2$ through \eqref{Q2 estimation}.
\end{itemize}

\STATE Obtain the estimate of the LR matrix: $\widehat{\bL} = \widehat{\bU}_p [\widehat{\bQ}_1,\widehat{\bQ}_2 ]\bP$.
\STATE Output: Estimate result $\widehat{\bL}$.
\end{algorithmic}
\end{algorithm}
\subsection{Coefficient matrix learning}
Given the expression of $\bL_1$ in \eqref{expression L1}, we denote $\bL_2 \in \R^{M \times (N-m)}$ as the sub-matrix of $\bL$ which collects the remaining $(N-m)$ columns of $\bL$, such that
$ \bL_2 = \bL \bZ_2$,
 where $\bZ_2 \in \R^{N \times (N-m)}$ is the column sampling matrix. Then, we have the following relationship between $[\bL_1, \bL_2]$ and $\bL$,
\beq
[\bL_1, \bL_2] = \bL [\bZ_1, \bZ_2]. \label{re 1 2}
\eeq
Given the expression in \eqref{svd of Ls}, \eqref{re 1 2} and the fact that the column subspace of $\bL$ is the same as that of $\bL_1$, thus, $ [\bL_1, \bL_2] $ can be represented as
\beq
[\bL_1 , \bL_2] = {\bU_p} \bQ
\eeq
where $\bQ \in \R^{r \times N}$ is the coefficient matrix of $[\bL_1 , \bL_2]$  associated with the column subspace matrix ${\bU}_p$.

\begin{remark}
For the coefficient matrix $\bQ$, we can rewrite $\bQ=[\bQ_1, \bQ_2]$, where $\bQ_1 \in \R^{r \times m}$ and $\bQ_2 \in \R^{r \times (N-m)}$ are the coefficient matrices of $\bL_1$ and $\bL_2$, respectively.
Therefore, the matrix $\bL$ can be expressed as
\beq
\bL &= &[\bL_1, \bL_2][\bZ_1, \bZ_2]^{-1}\nonumber \\
&= &[\bL_1, \bL_2]\bP\nonumber \\
&= &\bU_p [\bQ_1, \bQ_2] \bP, \label{true per}
\eeq
where $\bP = [\bZ_1, \bZ_2]^{-1} \in \R^{N \times N} $ can be treated as the permutation matrix and it is invertable.
Since we have already observed the $m$ columns of $\bL$,
we can get the estimation of $\bQ_1$ from the first step as
\beq
\widehat{\bQ}_1 = \widehat{\bU}_p^T \bY_1 = \widehat{\bU}_p^T\bL_1 +\sqrt{{mM}/{P_1}}\widehat{\bU}_p^T \bW_1. \label{Q1 estimation}
\eeq
Therefore, we only focus on the estimation for $\bQ_2$ in the second step.
\end{remark}

By denoting the affine map in the second step is $\cA_2(\cdot)$ and $p_2$ is the number of observations in the second step, the observation $\by_2 \in \R^{p_2}$ in the second step can be expressed as
\beq
\by_2 &=& \cA_2(\bU_p\bQ_2)+\bn_2 \nonumber \\
 &= & \bS_2 (\bI \otimes \bU_p) \vec(\bQ_2)+\bn_2, \label{ob second}
\eeq
where $\bn_2 \in \R^{p_2}$ with $\mathbb{E}(\bn_2 \bn_2^T)=\bC_2$, and $\bS_2 \in \R^{p_2 \times MN}$ is generated by reforming $\cA_2(\cdot)$. In the following part, we will discuss how to obtain the estimation for $\bQ_2$ and the design of the affine map $\cA_2(\cdot)$.

Given the estimated $\widehat{\bU}_p$ from the first step, we assume the expression of $\widehat{\bL}_2$ as
$
\widehat{\bL}_2 = \widehat{\bU}_p \widehat{\bQ}_2,
$
where $\widehat{\bQ}_2 \in \R^{r \times (N-m)}$.
Based on the analysis in Section \ref{algorithm}, under the observations $\by_2$ in \eqref{ob second}, the  efficient estimation of $\widehat{\bQ}_2$ is given by
\beq
\vec(\widehat{\bQ}_2) = (\widehat{\bA}_2^T \bC_2^{-1} \widehat{\bA}_2)^{-1} \widehat{\bA}_2^T \bC_2^{-1}\by_2, \label{Q estimate}
\eeq
where $\widehat{\bA}_2=\bS_2 (\bI \otimes \widehat{\bU}_p) \in \bR^{p_2 \times r(N-m)}$.
Here, in the second step, we assume the difference between $\text{col}(\widehat{\bU}_p)$ and $\text{col}(\bU_p)$ is small. Ideally, when $\text{col}(\widehat{\bU}_p)=\text{col}(\bU_p)$, the affine map $\bS_2$ with the minimal MSE can be obtained through the optimality condition in \eqref{general optimal condition1 c}.
Specifically, when $\mathbb{E}(\bn \bn^T)=\sigma^2 \bI_{p_2}$, the optimal $\bS_2$ associated with affine map $\cA_2(\cdot)$ is given by
\beq
\bS_2 = \sqrt{\frac{P_2}{r(N-m)} }
\bI_{r(N-m)}
 \left(  \bI \otimes \widehat{{\bU}}_p \right)^T. \label{optimal S2}
\eeq
Though the design of the affine map above utilize the approximation $\text{col}(\widehat{\bU}_p) \approx \text{col}(\bU_p)$, in the simulation part, we  verify that this approximation is still helpful to improve the reconstruction accuracy.
 Given the affine map in the design in \eqref{optimal S2}, for the coefficient matrix $\widehat{\bQ}_2$, we will have the following estimation by substituting \eqref{optimal S2} into \eqref{Q estimate},
\beq
 \widehat{\bQ}_2 = \widehat{\bU}_p^T \bL_2 + \sqrt{{r(N-m)}/{P_2}} \bW_2
= \widehat{\bU}_p^T\bL_2 +\bW_r,
\label{Q2 estimation}
 \eeq
where $\bW_2 \in \R^{r \times (N-m)}$ is the noise matrix which is reshaped from $\bn_2$ in \eqref{ob second}. The details of the proposed two-step method is stated in Algorithm 1.
The estimation result is given by
$\widehat{\bL} = \widehat{\bU}_p [\widehat{\bQ}_1,\widehat{\bQ}_2 ]\bP$, where $\bP$ is the permutation matrix defined in \eqref{true per}.

\subsection{Estimation error analysis}

Before evaluating the estimation error, the following Lemma provides the accuracy of the estimated column subspace $\text{col}(\widehat{\bU}_p)$ compared to $\text{col}(\bU_p)$. Denote $\bU_{p\bot} \in \R^{M \times (N-r)}$ as the matrix which satisfies $\bU_{p\bot} \bU_{p\bot}^T  = \bI_M - \bU_{p} \bU_{p}^T$.
We introduce the
\beq
\eta=\lA  \widehat{\bU}_{p\bot}^T \bU_p  \rA_2 =  \lA  \widehat{\bU}_{p}^T \bU_{p\bot}  \rA_2 \label{subspace metric}
\eeq
as the subspace distance defined in \cite{matrixcom}. The value of $\eta$ is ranged from $0$ to $1$. When $\text{col}(\widehat{\bU}_p)=\text{col}({\bU}_p)$, we can easily verify $\eta = 0$.

\begin{lemma} \label{subspace bound lemma}

Assume the singular value gap $\delta = \lam_{r} - \widehat{\lam}_{r+1}$, where $\lam_{r}$ is the $r$th largest singular value of $\bL_1 \in \R^{M \times m}$, and  $\widehat{\lam}_{r+1}$ is the $(r+1)$th largest singular value of $\bY_1 \in \R^{M \times m}$ with $\bY_1 = \bL_1 + \bW_s$. According to the Wedin's Theorem \cite{li1998relative, stewart2006perturbation}, the following bound holds
\beq
\lA  \widehat{\bU}_{p}^T \bU_{p\bot}  \rA_2 \!\!& \le &\frac{\max \{\lA \bW_s \bV_p\rA_2, \lA \bW_s^T \bU_p\rA_2 \}}{\delta} \nonumber\\
 \!\! &\le& \!\! \frac{\lA \bW_s \rA_2}{\delta}  \nonumber\\
  \!\! &=& \!\! \sqrt{\frac{m M}{P_1}}\ \frac{\lA \bW_1 \rA_2}{\delta}, \label{subspace distance bound}
\eeq
where $\bY_1 = \widehat{\bU}_p \widehat{\bSig}_p \widehat{\bV}_p^T $ and $\bL_1 = {\bU}_p {\bSig}_p {\bV}_p^T $ are the SVD of $\bY_1$ and $\bL_1 $, respectively.
\end{lemma}

Given the subspace estimation accuracy in \eqref{subspace distance bound}, it is of interest of evaluate estimation error between $\bL$ and $\widehat{\bL}$, which is given in the following lemma.
\begin{lemma}
Suppose the low-rank matrix $\bL \in \R^{M \times N}$, and the estimation result $\widehat{\bL}$ is obtained through Algorithm 1. The estimation error is given by
\beq
\begin{aligned}
\lA \bL - \widehat{\bL} \rA_F \le
 \!\! \sqrt{\frac{mM}{P_1}}\!\! \left(\!\!\lA  \bW_1 \rA_F \!\!+\!\!\frac{\lA \bW_1 \rA_2\lA \bL \rA_F}{\delta}\!\! \right)\!\! + \!\!\sqrt{\!\!\frac{r(N-m)}{P_2}}\!\! \lA \bW_2 \rA_F ,
\label{error bound2}
\end{aligned}
\eeq
where $\delta$ is defined in same way as Lemma \ref{subspace bound lemma}.
\end{lemma}
\begin{proof}
Given the expression of $\widehat{\bQ}$, we formulate the estimation error in the following form,
\beq
\begin{aligned}
\lA \bL- \widehat{\bU}_p \widehat{\bQ}  \bP \rA_F
&\le \lA \widehat{\bU}_p \widehat{\bU}_p^T  \bW_s \!\!+\!\! \bW_r \rA_F\!\! +\!\! \lA \bL - \widehat{\bU}_p \widehat{\bU}_p^T \bL \rA_F\\
&\le \lA \widehat{\bU}_p \widehat{\bU}_p^T  \bW_s \rA_F\!\! +\!\!  \lA \bW_r \rA_F  \!\!+\!\! \lA \bL - \widehat{\bU}_p \widehat{\bU}_p^T \bL \rA_F\\
 &\le \lA   \bW_s \rA_F\!\! +\!\!  \lA \bW_r \rA_F \!\!+ \!\!\lA \left( \bI\!\! - \!\! \widehat{\bU}_p \widehat{\bU}_p^T\right) \bU_p \bU_p^T \rA_2 \!\! \lA \bL \rA_F\\
 &= \lA   \bW_s \rA_F \!\!+  \lA \bW_r \rA_F\!\! +  \!\! \lA \left(\widehat{\bU}_{p\bot} \widehat{\bU}_{p\bot}^T \right) \!\! \bU_p \bU_p^T \rA_2  \lA \bL \rA_F\\
 & \le  \lA  \bW_s \rA_F +  \lA \bW_r \rA_F + \lA  \widehat{\bU}_{p\bot}^T \bU_p  \rA_2 \lA \bL \rA_F, \label{error1}
\end{aligned}
\eeq
where $\widehat{\bU}_{p\bot} \widehat{\bU}_{p\bot}^T$ denotes the projection matrix onto the complementary subspace spanned by $\widehat{\bU}_p$, and $\bU_p \bU_p^T$ is the projection matrix onto the subspace spanned by $\bU_p$.

Therefore, combining \eqref{error1} and \eqref{subspace distance bound}, we can have the error in \eqref{error bound2}.
This concludes the proof.\hfill \qed
\end{proof}

Note that the bound in \eqref{error bound2} will hold for any possible noise.
Compared to the MSE when the column subspace is exact, the MSE in \eqref{error bound2} includes one more part, such that $\sqrt{\frac{m M}{P_1}}\frac{\lA \bW_1 \rA_2\lA \bL \rA_F}{\delta}$, which comes from the error of estimate of column subspace.
Using the proposed two-step method, the necessary number of samples is
\beq
p=m M + r(N-m).
\eeq
In particular, if we set the $m=r$, the number of observations is just equal to the degrees of freedom of esimated low-rank matrix.
\begin{remark}
The complexity of the two-step method mainly comes from SVD computation $\bY_1$, i.e., $\mathcal{O} (Mm^2)$. Thus, when $m$ is just several times of the rank $r$, the computational complexity of the proposed two-step method is in a small order.
\end{remark}

\begin{figure}[htbp]
\centering
\includegraphics[width=3.2 in]{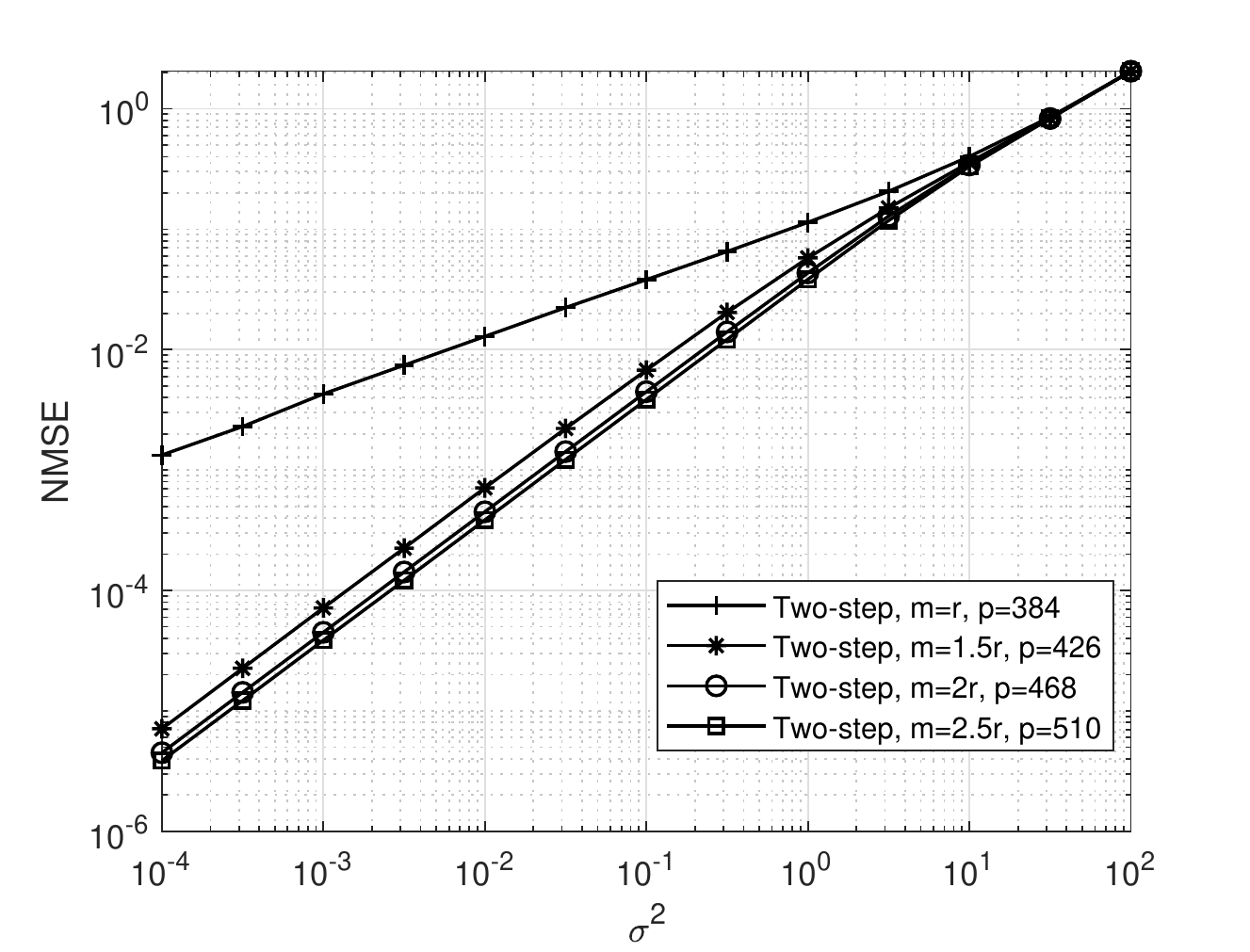}
\caption{NMSE  vs. Noise levels with differen number of observations($M = 20, N = 50, r=6$).}  \label{different observations}
\end{figure}
\section{Simulation results} \label{simulation}

In this section, we simulate the performance of the proposed two-step matrix reconstruction method compared to two non-adaptive methods, i.e., MF \cite{AlternatingMin} and NNM \cite{recht}, and one adaptive method, i.e., subspace pursuit (SP) approach\cite{HighDimension}.
The simulation parameters are $M=20, N=50, r=6$. The normalized MSE (NMSE) is defined as
\beq
\text{NMSE} = \E\left( \lA\widehat{\bL} - \bL \rA_F^2 /\lA \bL \rA_F^2 \right). \nonumber
\eeq

Here, we assume the power associate with the two steps are $P_1=P_2 = MN$.
For the simulated methods, each point of the curve is plotted by averaging the NMSE of 1000 trails.
\subsection{The affect of number of observations }
First of all, we illustrate the effect of number of observations on the reconstruction accuracy in Fig. \ref{different observations}. We let the number of column $m$ for the first step is ranged in $\{ r,1.5r,2r,2.5r\} = \{6,9,12,15 \}$. Accordingly, the total number of observations are
$mM + r(N-m)$.
In general, when more columns are utilized for the first step, more accurate columns subspace can be obtained.
As we can see, when the number of columns for the first step is $1.5r$, the robust performance can be achieved. This means that when we let $m = 1.5r$ in the first step, we can acquire a robust estimation for the subspace information.
\begin{figure}[htbp]
\centering
\includegraphics[width=3.2 in]{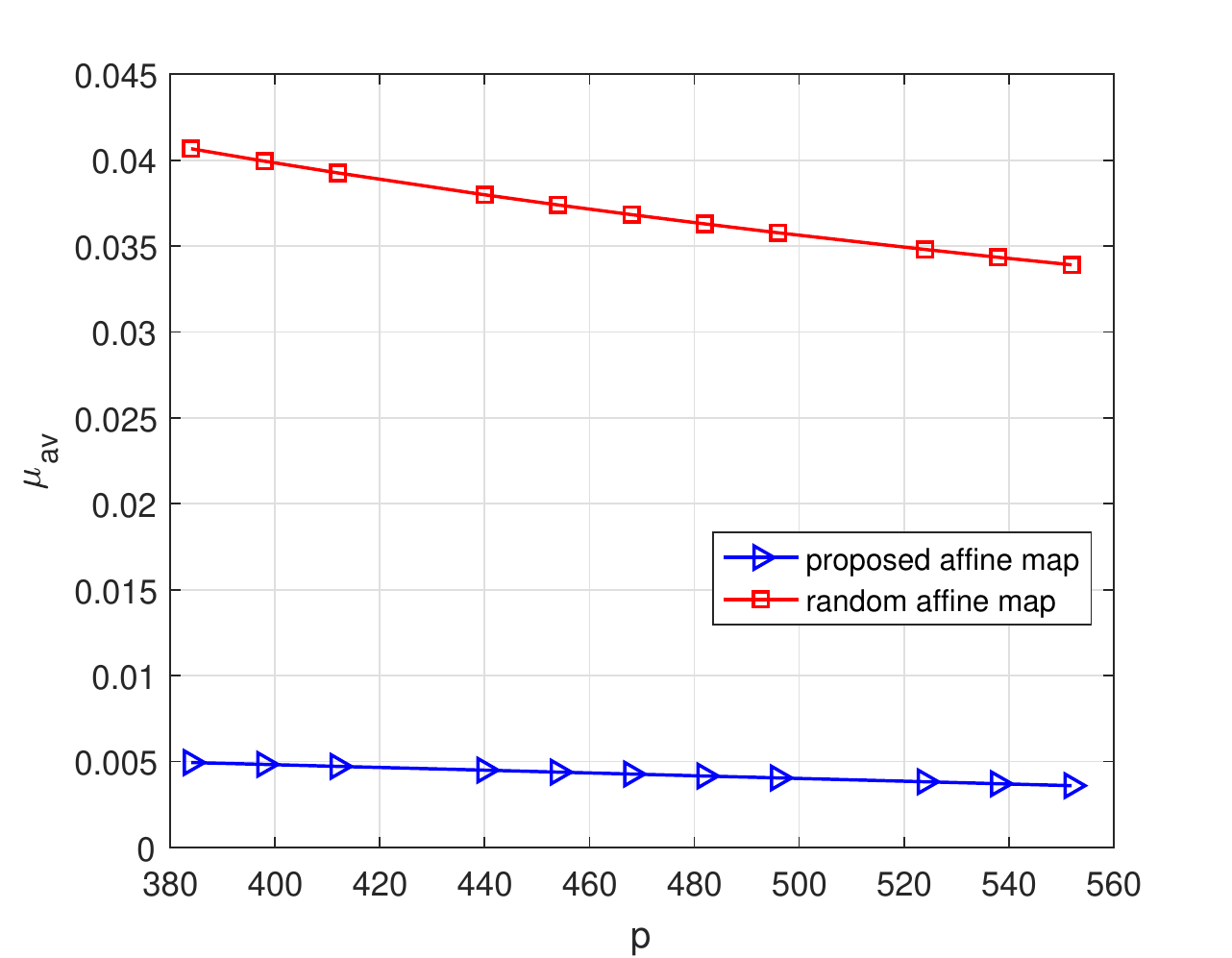}
\caption{Averaged mutual coherence  vs. Number of observations   ($M = 20, N = 50$).} \label{figure_RIP}
\end{figure}
\redd{
\subsection{The characteristic of designed affine map}}
\redd{In this subsection, we compare the affine map of the proposed two-step method with the randomly generated affine map. As far as we know, it is NP-hard to compute the RIP constant for the given affine map. Therefore, we turn to compute the alternative constant to RIP, namely, averaged mutual coherence defined in \cite{elad_CS}, which evaluates the averaged coherence of the columns of affine map matrix $\bS \in \R^{p \times MN}$. Specifically,  letting $\bar{\bS}\in \R^{p \times MN}$ be the matrix which normalizes the columns of $\bS$, the averaged mutual coherence of $\bS$ is given by
\beq
\mu_{av}(\bS) = \frac{\sum \limits_{1 \le i,j \le MN,i\neq j} \tr([\bar{\bS}]_{:,i}^T [\bar{\bS}]_{:,j})}{M^2 N^2-p}. \nonumber
\eeq
It has been shown that the reconstruction accuracy is related to the value of $\mu_{av}$ \cite{elad_CS, matrixOptimization}. In Fig. \ref{figure_RIP}, we evaluate the averaged mutual coherence of proposed affine map compared to that of gaussian random affine map, i.e., $[\bS]_{i,j} \sim \cN(0,1), \forall i,j$.
As can be seen from Fig. \ref{figure_RIP}, the proposed affine map design achieves the lower averaged mutual coherence compared to the randomly generated affine map. Therefore, it could be expected that reconstruction accuracy achieved by the proposed affine map will be higher than that of randomly generated affine map.}

\subsection{The accuracy comparisons with benchmarks}
\begin{figure}[htbp]
\centering
\includegraphics[width=3.2 in]{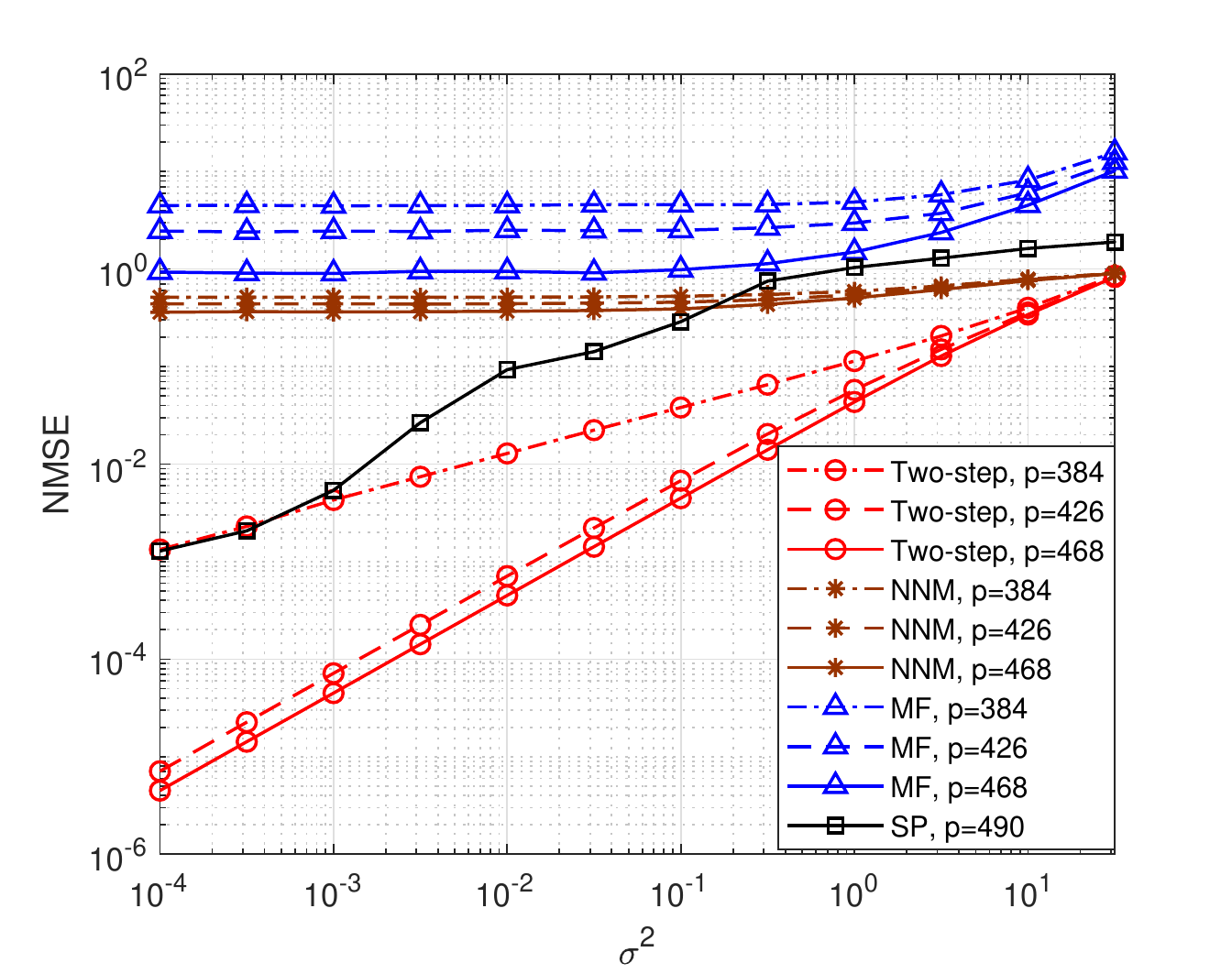}
\caption{ NMSE vs. Noise levels  ($M = 20, N = 50, r=6$).} \label{figure_benchmark}
\end{figure}

\begin{figure}[htbp]
\centering
\includegraphics[width=3.2 in]{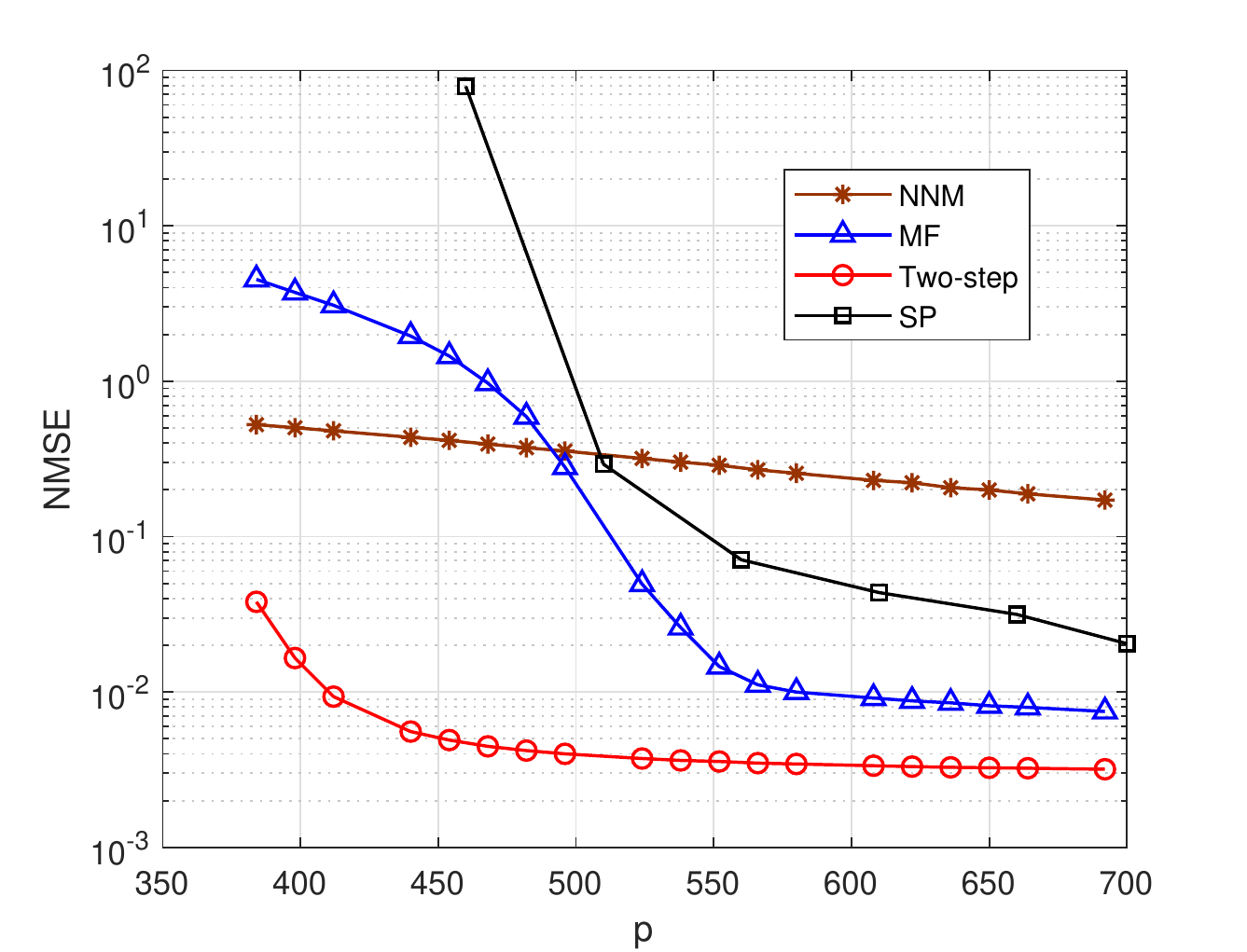}
\caption{ NMSE vs. Number of observations   ($M = 20, N = 50, r=6,\sigma^2=0.1$).} \label{benchmark_different_observations}
\end{figure}

In Fig. \ref{figure_benchmark}, we compare the performance of proposed two-step method with NNM, MF and SP approaches. We set the number of observations for proposed two-step method, NNM and MF selected from  $\{384,426,468 \}$. Since the SP method requires more observations in order to achieve valid reconstruction, we let the number of observations for SP be $p=490$.
As can been seen that NNM, MF and proposed two-step method can benefit from the increasing of number of observations.
It is clear that the proposed two-step method outperforms the NNM, MF and SP. Moreover, the performance gap is larger when the noise level is low. In particular, when the number of observations is $p=426$, the proposed two-step method has a clear performance gap compared to the benchmarks. This is because the performance of NNM, MF and SP will be restricted by the number of observations. If the number of observations is not sufficient, i.e., $p=426$, the accuracy of NNM, MF and SP will be saturated when the noise level is low. Different from the saturated phenomenon of benchmarks, for proposed two-step method, more accurate reconstruction can be obtained if the noise level is decreasing furthermore.

In Fig. \ref{benchmark_different_observations}, we evaluate the performance of proposed two-step method under different observations compared to the MF,  NNM, and SP. The noise level is set as  $\sigma^2=0.1$.
As we can see, all the simulated methods will benefits from the increasing of the number of observations.
Moreover, the proposed two-step method outperforms the others. When the number of observations is $p=384$, which is equal to the degrees of freedom of the simulated low-rank matrix, i.e., $(50+20-6)\times 6=384$, the proposed two-step method can achieve a robust estimation. This validates the fact the proposed two-step method only requires the number of observations that is approximately equal to the degrees of freedom of low-rank matrix.
\begin{figure}[htbp]
\centering
\includegraphics[width=3.2 in]{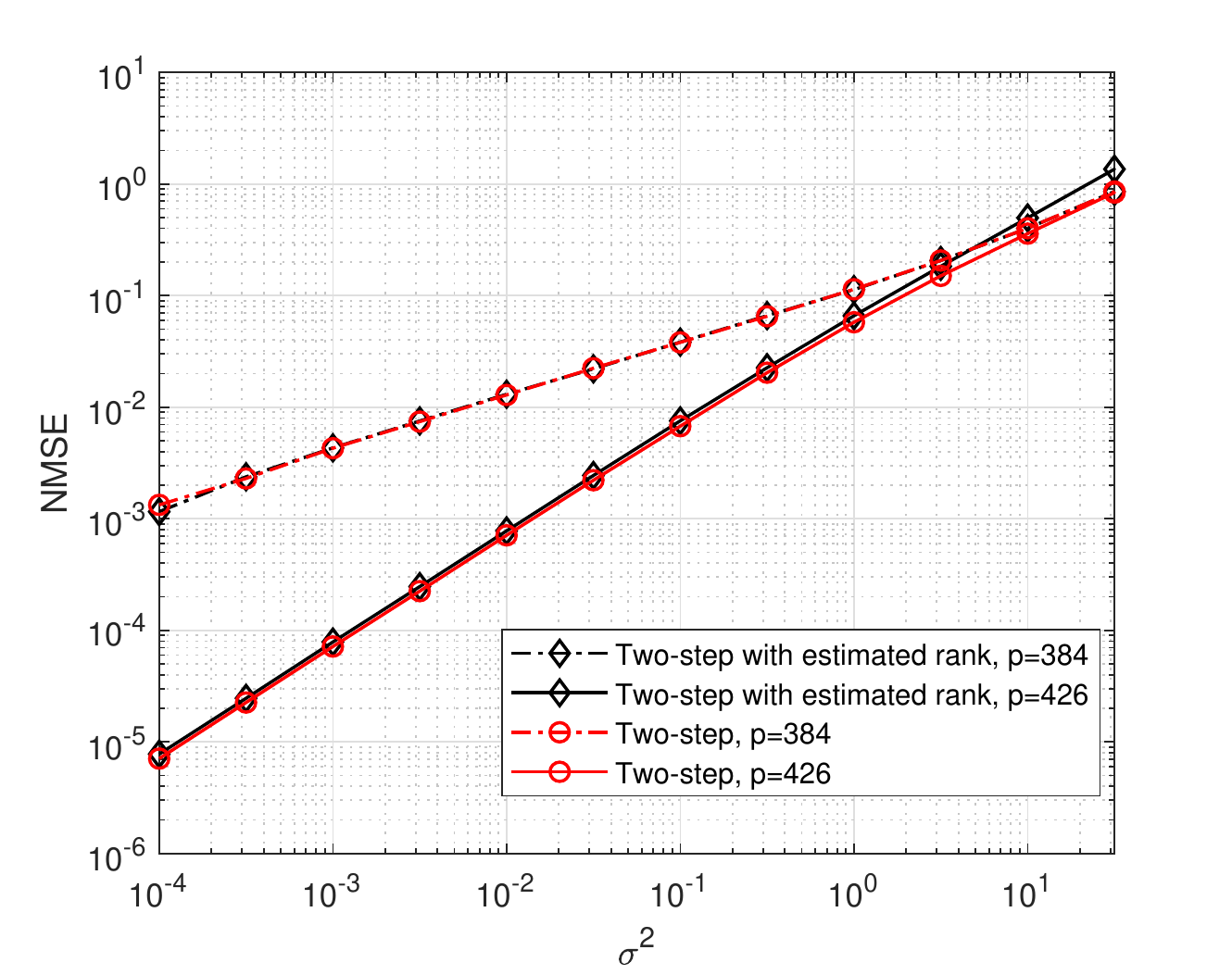}
\caption{ NMSE vs. Noise levels   ($M = 20, N = 50, r=6$).} \label{estimated r with true r}
\end{figure}
\subsection{The reconstruction accuracy by using estimated rank information}
In this simulation part, we simulate the reconstruction performance when we utilize the estimated rank information instead of true rank information. As we can see in Fig. \ref{estimated r with true r}, by using the estimated rank $\hat{r}$ in \eqref{rank estimate}, the similar performance can be achieved compared to the scenario when we utilize the true rank information. Therefore, the rank estimation method in \eqref{rank estimate} is robust for the low-rank matrix reconstruction under the noisy cases.

\section{Conclusion} \label{conclusion}
In this paper, we investigate the low-rank matrix reconstruction when the subspace information is known.
Under the subspace information, the optimal representation of low-rank matrix is analyzed in order to obtain minimal MSE. In the case that no subspace information is aware, the proposed two-step method can handle this practical scenario. The first step will obtain the column subspace of the low-rank matrix, and the second step will get the remaining coefficient information of the low-rank matrix. By using the observations approximately equal to degrees of freedom of the low-rank matrix, the simulation results show that the proposed two-step method experiences robust performance compared to the existing reconstruction methods.

\section*{Acknowledgments}
This work is supported in part by a grant from the Research Grants Council of the Hong Kong SAR, China (No. CityU 11272516), and by Theme Based Research (No. T42-103/16-N).

\section*{References}

\bibliography{reference}

\begin{thebibliography}{10}
\expandafter\ifx\csname url\endcsname\relax
  \def\url#1{\texttt{#1}}\fi
\expandafter\ifx\csname urlprefix\endcsname\relax\def\urlprefix{URL }\fi
\expandafter\ifx\csname href\endcsname\relax
  \def\href#1#2{#2} \def\path#1{#1}\fi

\bibitem{SpareRecovery}
M.~E. Davies, Y.~C. Eldar, Rank awareness in joint sparse recovery, IEEE Trans.
  Inf. Theory 58~(2) (2012) 1135--1146.
\newblock \href {http://dx.doi.org/10.1109/TIT.2011.2173722}
  {\path{doi:10.1109/TIT.2011.2173722}}.

\bibitem{liu2009interior}
Z.~Liu, L.~Vandenberghe, Interior-point method for nuclear norm approximation
  with application to system identification, SIAM J. Matrix Anal. Appl. 31~(3)
  (2009) 1235--1256.

\bibitem{recht}
B.~Recht, M.~Fazel, P.~A. Parrilo, Guaranteed minimum-rank solutions of linear
  matrix equations via nuclear norm minimization, SIAM Rev. 52~(3) (2010)
  471--501.

\bibitem{chen2015exact}
Y.~{Chen}, Y.~{Chi}, A.~J. {Goldsmith}, Exact and stable covariance estimation
  from quadratic sampling via convex programming, IEEE Trans. Inf. Theory
  61~(7) (2015) 4034--4059.
\newblock \href {http://dx.doi.org/10.1109/TIT.2015.2429594}
  {\path{doi:10.1109/TIT.2015.2429594}}.

\bibitem{cai2015rop}
T.~T. Cai, A.~Zhang, Rop: Matrix recovery via rank-one projections, The Annals
  of Statistics 43~(1) (2015) 102--138.

\bibitem{AlternatingMin}
P.~Jain, P.~Netrapalli, S.~Sanghavi, Low-rank matrix completion using
  alternating minimization, in: Proceedings of the Forty-fifth Annual ACM
  Symposium on Theory of Computing, STOC '13, ACM, New York, NY, USA, 2013, pp.
  665--674.
\newblock \href {http://dx.doi.org/10.1145/2488608.2488693}
  {\path{doi:10.1145/2488608.2488693}}.

\bibitem{candes2011tight}
E.~J. {Candes}, Y.~{Plan}, Tight oracle inequalities for low-rank matrix
  recovery from a minimal number of noisy random measurements, IEEE Trans. Inf.
  Theory 57~(4) (2011) 2342--2359.
\newblock \href {http://dx.doi.org/10.1109/TIT.2011.2111771}
  {\path{doi:10.1109/TIT.2011.2111771}}.

\bibitem{database}
D.~Achlioptas, Database-friendly random projections: Johnson-lindenstrauss with
  binary coins, Journal of computer and System Sciences 66~(4) (2003) 671--687.

\bibitem{ZhangSD}
W.~Zhang, T.~Kim, D.~J. Love, E.~Perrins, Leveraging the restricted isometry
  property: Improved low-rank subspace decomposition for hybrid millimeter-wave
  systems, IEEE Trans. Commun. 66~(11) (2018) 5814--5827.
\newblock \href {http://dx.doi.org/10.1109/TCOMM.2018.2854779}
  {\path{doi:10.1109/TCOMM.2018.2854779}}.

\bibitem{SVP}
P.~Jain, R.~Meka, I.~S. Dhillon, Guaranteed rank minimization via singular
  value projection, in: Advances in Neural Information Processing Systems,
  2010, pp. 937--945.

\bibitem{IHT_matrix}
J.~Tanner, K.~Wei, Normalized iterative hard thresholding for matrix
  completion, SIAM J. Sci. Comput. 35~(5) (2013) S104--S125.

\bibitem{subspaceAware}
S.~Biswas, S.~Dasgupta, R.~Mudumbai, M.~Jacob, Subspace aware recovery of low
  rank and jointly sparse signals, IEEE Trans. Comput. Imaging 3~(1) (2017)
  22--35.
\newblock \href {http://dx.doi.org/10.1109/TCI.2016.2628352}
  {\path{doi:10.1109/TCI.2016.2628352}}.

\bibitem{HighDimension}
M.~Rahmani, G.~K. Atia, High dimensional low rank plus sparse matrix
  decomposition, IEEE Trans. Signal Process. 65~(8) (2017) 2004--2019.
\newblock \href {http://dx.doi.org/10.1109/TSP.2017.2649482}
  {\path{doi:10.1109/TSP.2017.2649482}}.

\bibitem{overview}
M.~A. Davenport, J.~Romberg, An overview of low-rank matrix recovery from
  incomplete observations, IEEE J. Sel. Top. Signal Process. 10~(4) (2016)
  608--622.
\newblock \href {http://dx.doi.org/10.1109/JSTSP.2016.2539100}
  {\path{doi:10.1109/JSTSP.2016.2539100}}.

\bibitem{proximal}
N.~Parikh, S.~P. Boyd, et~al., Proximal algorithms., Foundations and Trends in
  Optimization 1~(3) (2014) 127--239.

\bibitem{combettes2011proximal}
P.~L. Combettes, J.-C. Pesquet, Proximal splitting methods in signal
  processing, in: Fixed-point Algorithms for Inverse Problems in Science and
  Engineering, Springer, 2011, pp. 185--212.

\bibitem{cai2010}
J.-F. Cai, E.~J. Cand{\`e}s, Z.~Shen, A singular value thresholding algorithm
  for matrix completion, SIAM J. Optim. 20~(4) (2010) 1956--1982.

\bibitem{accYao}
Q.~{Yao}, J.~T. {Kwok}, Accelerated and inexact soft-impute for large-scale
  matrix and tensor completion, IEEE Trans. Knowl. Data Eng. (2018) 1--1\href
  {http://dx.doi.org/10.1109/TKDE.2018.2867533}
  {\path{doi:10.1109/TKDE.2018.2867533}}.

\bibitem{YaoNonCVX}
Q.~{Yao}, J.~T. {Kwok}, T.~{Wang}, T.~{Liu}, Large-scale low-rank matrix
  learning with nonconvex regularizers, IEEE Trans. Pattern Anal. Mach. Intell.
  (2018) 1--1\href {http://dx.doi.org/10.1109/TPAMI.2018.2858249}
  {\path{doi:10.1109/TPAMI.2018.2858249}}.

\bibitem{koren2009matrix}
Y.~Koren, R.~Bell, C.~Volinsky, Matrix factorization techniques for recommender
  systems, Computer 42~(8).

\bibitem{Haldar_factor}
J.~P. {Haldar}, D.~{Hernando}, Rank-constrained solutions to linear matrix
  equations using powerfactorization, IEEE Signal Process Lett. 16~(7) (2009)
  584--587.
\newblock \href {http://dx.doi.org/10.1109/LSP.2009.2018223}
  {\path{doi:10.1109/LSP.2009.2018223}}.

\bibitem{ZhangSparse}
W.~Zhang, T.~Kim, D.~Love, Sparse subspace decomposition for millimeter wave
  {MIMO} channel estimation, in: 2016 IEEE Global Communications Conference:
  Signal Processing for Communications (Globecom2016 SPC), Washington, USA,
  2016.

\bibitem{estimationBook}
S.~M. Kay, Fundamentals of statistical signal processing: Practical algorithm
  development, Vol.~3, Pearson Education, 2013.

\bibitem{boyd}
S.~Boyd, L.~Vandenberghe, Convex Optimization, Cambridge University Press,
  2004.

\bibitem{RandomMatrix}
G.~W. Anderson, A.~Guionnet, O.~Zeitouni, An introduction to random matrices,
  volume 118 of cambridge studies in advanced mathematics (2010).

\bibitem{matrixcom}
G.~Golub, C.~Van~Loan, Matrix Computations, Johns Hopkins Studies in the
  Mathematical Sciences, Johns Hopkins University Press, 2013.

\bibitem{li1998relative}
R.~Li, Relative perturbation theory: Ii. eigenspace and singular subspace
  variations, SIAM J. Matrix Anal. Appl. 20~(2) (1998) 471--492.
\newblock \href {http://dx.doi.org/10.1137/S0895479896298506}
  {\path{doi:10.1137/S0895479896298506}}.

\bibitem{stewart2006perturbation}
M.~Stewart, Perturbation of the svd in the presence of small singular values,
  Linear Algebra Appl. 419~(1) (2006) 53 -- 77.
\newblock \href {http://dx.doi.org/https://doi.org/10.1016/j.laa.2006.04.013}
  {\path{doi:https://doi.org/10.1016/j.laa.2006.04.013}}.

\bibitem{elad_CS}
M.~{Elad}, Optimized projections for compressed sensing, IEEE Trans. Signal
  Process. 55~(12) (2007) 5695--5702.
\newblock \href {http://dx.doi.org/10.1109/TSP.2007.900760}
  {\path{doi:10.1109/TSP.2007.900760}}.

\bibitem{matrixOptimization}
G.~{Li}, Z.~{Zhu}, D.~{Yang}, L.~{Chang}, H.~{Bai}, On projection matrix
  optimization for compressive sensing systems, IEEE Trans. Signal Process.
  61~(11) (2013) 2887--2898.
\newblock \href {http://dx.doi.org/10.1109/TSP.2013.2253776}
  {\path{doi:10.1109/TSP.2013.2253776}}.

\end{thebibliography}
\end{document}